\documentclass[tbtags,reqno]{amsart}
\usepackage{graphicx,color}
\usepackage{amsmath,amsfonts,amstext,amsbsy,amsopn,amsthm,amscd,amsxtra,dsfont,amssymb,pb-diagram}
\usepackage{amsmath,amsthm,amsopn,amsfonts,amssymb,epsfig,enumerate}
\usepackage{amsthm,amsopn,amsfonts,amssymb,epsfig,mathrsfs,wasysym}
\usepackage[all]{xy}
\usepackage[active]{srcltx}
\usepackage{xcolor}

\numberwithin{equation}{section}

\makeatletter
\renewcommand{\subsubsection}{\@startsection
{subsubsection} {3} {0mm} {\baselineskip} {-0.5\baselineskip} {\normalfont\normalsize\bfseries}} \makeatother

\newtheorem{theorem}{Theorem}[section]
\newtheorem{lemma}[theorem]{Lemma}
\newtheorem{proposition}[theorem]{Proposition}

\newtheorem{corollary}[theorem]{Corollary}

\newtheorem{remark}[theorem]{Remark}
\newtheorem*{acknow}{Acknowledgments}

\def\A{A}
\def\osp{{\mathfrak{osp}}}
\def\bi{\mathcal{BI}}

\begin{document}

\title[$\osp(1,2)$ and generalized Bannai-Ito algebras]
{$\osp(1,2)$ and generalized Bannai-Ito algebras}

\author{Vincent X. Genest }
\address{Department of Mathematics, Massachussetts Institute of Technology, Cambridge, MA 02139, USA} \email{vxgenest@mit.edu}

\author{Luc Lapointe}
\address{Instituto de Matem\'atica y F\'{\i}sica, Universidad de Talca, Casilla 747, Talca,
Chile} \email{lapointe@inst-mat.utalca.cl}

\author{Luc Vinet}
\address{Centre de Recherches Math\'ematiques, Universit\'e de Montr\'eal, Montr\'eal, QC H3C 3J7, Canada} \email{luc.vinet@umontreal.ca}


\begin{abstract}   Generalizations of the (rank 1)
 Bannai-Ito algebra are obtained from a refinement of the grade involution of the Lie superalgebra $\osp(1,2)$.
  A hyperoctahedral extension is derived  by using a realization of $\osp(1,2)$  in terms of Dunkl operators associated to the Weyl group $B_3$.
\end{abstract}


\maketitle

\section{Introduction}

The Bannai-Ito algebra $\bi_3$ is the associative algebra over $\mathbb C$ with three generators $K_{12}, K_{23}, K_{13}$ that satisfy the relations
\begin{equation} \label{eq1.1}
\{K_{12}, K_{23}\}=K_{13}+\omega_{13}, \qquad \{K_{12}, K_{13}\}=K_{23}+\omega_{23},\qquad \{K_{13}, K_{23}\}=K_{12}+\omega_{12}  
\end{equation}
where $\{A,B \}=AB+BA$ is the anticommutator and where $\omega_{12}$, $\omega_{23}$ and $\omega_{13}$ are structure constants.  It is readily verified that the Casimir element
\begin{equation}
Q=K_{12}^2+ K_{13}^2+ K_{23}^2
\end{equation}
belongs to the center of $\bi_3$.

The Bannai-Ito algebra was first presented in \cite{cit1} as the algebra encoding the bispectral properties of the Bannai-Ito polynomials \cite{cit2}. Its connection to the Lie superalgebra $\osp(1,2)$ was understood soon after when (a central extension of) $\bi_3$ was shown to be the centralizer of the coproduct embedding of $\osp(1,2)$ in the three-fold product of this superalgebra \cite{cit3,cit4,cit5}.  As a consequence,   the Bannai-Ito polynomials were seen \cite{cit3} to be essentially the Racah coefficients of $\osp(1,2)$.  In the following, we shall keep the notation $\bi_3$ when the symbols $\omega_{ij}$, $ij= 12,13,23$, are central elements instead of constants.

The Bannai-Ito algebra was further shown \cite{cit6} to be isomorphic to the degenerate double affine Hecke algebra of type $(C_1^{\vee},C_1)$.

This Bannai-Ito algebra proves relevant in a variety of contexts where realizations of $\osp(1,2)$ arise. It intervenes in Dunkl harmonic analysis on the 2-sphere \cite{cit7} and is the symmetry algebra in three dimensions of a superintegrable model with reflections \cite{cit8} and of the Dirac-Dunkl equation \cite{cit9}.  These models are actually quite useful in the identification of the higher rank extensions of $\bi_3$ that have been obtained \cite{cit10}.

In view of the fundamental features and various applications of the Bannai-Ito algebra, it should be of interest to explore possible generalizations.  When the Bannai-Ito algebra is identified as the algebra of the intermediate Casimir operators in the triple product of $\osp(1,2)$ algebras, the reflection group $\mathbb Z_2^3$ formed by the grade involutions of each factor in $\osp(1,2)^{\otimes 3}$ is present.  The purpose of this paper is to show that the Bannai-Ito algebra admits generalizations that lie outside the setting of coproduct homomorphisms when the $\osp(1,2)$ grade involution decomposes in a product of supplementary involutions that leave the even part of the algebra invariant.  The extended Bannai-Ito algebra will be generated by the elements centralizing $\osp(1,2)$ that can be found in this situation. The results bear a connection with recent work on the symmetries of the Dirac-Dunkl equation \cite{cit11}.

The paper will proceed as follows.  The definition and basic features of the Lie superalgebra $\osp(1,2)$ will be recalled in Section~\ref{sec2} where the supplementary involutions and their properties will be introduced.  Only three such involutions will be needed to extend
the rank~1 Bannai-Ito algebra.
The centralizing elements will be constructed in Section~\ref{sec3} and the relations defining the algebra they generate will be given.  They will be seen to extend \eqref{eq1.1} non-trivially in general.

Section~\ref{sec4} will explain how the standard Bannai-Ito algebra $\bi_3$ is recovered when the framework is specialized to the three-fold product of $\osp(1,2)$.  Section~\ref{sec5} will be dedicated to an interesting hyperoctahedral extension of the Bannai-Ito algebra that results when $\osp(1,2)$ is realized in terms of Dunkl operators associated to the $B_3$-Weyl group. Explicit expressions will be provided and an extension of \eqref{eq1.1} involving elements of the signed permutation group on three objects will be obtained. Section~\ref{sec6} will offer a discussion of the generalized Bannai-Ito algebra that is found when the $\osp(1,2)$ realization involves a Clifford algebra. The paper will end with concluding remarks. Some details of the derivation of the structure relations of the generalized algebra will be given in the Appendix.

\section{$\osp(1,2)$ and involutions} \label{sec2}

The Lie superalgebra $\osp(1,2)$ can be presented by adjoining to the three generators $\A_0$, $A_+$ $A_-$ the grade involution $P$ that accounts for the $\mathbb Z_2$-grading of the algebra.  That $\A_0$ and $A_{\pm}$ are respectively even and odd generators is enforced by taking
\begin{equation}
P^2= 1,\qquad [P,\A_0]=0,\qquad \{P,A_{\pm} \}=0
\end{equation}
with $[A,B]=AB-BA$.  This is then supplemented by the relations
\begin{equation} \label{eq2.2}
[\A_0 , A_{\pm}]=\pm A_{\pm},\qquad \{A_+,A_- \}=2 \A_0
  \end{equation}
to complete the definition of $\osp(1,2)$.  The Casimir element
\begin{equation}
S = \frac{1}{2} \bigl( [A_-,A_+]-1  \bigr)
\end{equation}
enjoys the same relations that $P$ does with $\A_0$ and $A_{\pm}$ and thus
$\Gamma=SP$ commutes with all generators:
\begin{equation}
[\Gamma, A_{\pm}]=[\Gamma, \A_0]=[\Gamma,P]=0 \, .
\end{equation}

\noindent Recalling that $[A,BC]= \{A,B \}C-B\{A,C\}$, it is useful to record that
\begin{equation} \label{eq2.5}
[A_{\mp}, A_{\pm}^2]= 2[\A_0,A_{\pm}]= \pm 2 A_{\pm} \, .
\end{equation}
The three even generators 
\begin{equation}
B_{\pm}= A_{\pm}^2 \quad \text{and} \quad \A_0
\end{equation}
are readily seen to close onto the $su(1,1)$ commutation relations and to commute with $P$:
\begin{equation} \label{eq2.7}
[B_+,B_-]=-4\A_0,\qquad [\A_0,B_{\pm}]= \pm 2B_{\pm},\qquad [P,B_{\pm}]=[P,\A_0]=0 \, .
\end{equation}
The $su(1,1)$ Casimir element
\begin{equation}
C= \frac{1}{4} \bigl ( \A_0^2 -B_+ B_- -2 \A_0 \bigr)
\end{equation}
is related as follows with the Casimir element $S=\Gamma P$:
\begin{equation}
\Gamma^2- \Gamma P = 4 C+\frac{3}{2}\, .
\end{equation}
Let us now introduce supplementary involutions $P_i$, $i =1,\dots,n$, such that
\begin{equation}
P= P_1 P_2 \cdots P_n
\end{equation}
and
\begin{equation}
P_i^2=1 ,\qquad [P_i,P_j]=0,\quad i\neq j \, .
\end{equation}  
We shall make two additional assumptions on the relations that these involutions have with the $\osp(1,2)$ generators:
\begin{enumerate}
\item[i.] We will suppose that they all commute with the $su(1,1)$ generators:
  \begin{equation}
[P_i,\A_0]= [P_i, B_{\pm}]=0,\quad \forall i
  \end{equation}
 \item[ii.] We shall impose an additivity or decomposition property of the form 
   \begin{equation} \label{eq2.13}
      \bigl[P_i, [ P_j,A_{\pm}] \bigr]=0, \quad i\neq j
   \end{equation}
   stating that the commutator of the odd generators with any involution is even with respect to all the others.
\end{enumerate}
This last property entails the following lemma that will prove quite useful.
\begin{lemma}
  Subject to the above definitions and hypotheses, together with the $\osp(1,2)$ generators $A_+$ and $A_-$, the involutions $P_i$, $i=1,\dots,n$, satisfy the following relations:
  \begin{equation} \label{eq2.14}
P_i A_{\pm} P_j + P_j A_{\pm} P_i = P_i P_j A_{\pm} +  A_{\pm} P_i P_j,\quad i\neq j 
  \end{equation}
  and
 \begin{equation} \label{eq2.15}
   P_i A_{\pm }A_{\mp} A_{\pm} P_j +  P_j A_{\pm }A_{\mp} A_{\pm} P_i =
    P_i P_j A_{\pm }A_{\mp} A_{\pm}  +  A_{\pm }A_{\mp} A_{\pm} P_i P_j, \quad i\neq j.
     \end{equation}
  \end{lemma}  
\begin{proof} Formula \eqref{eq2.14} straightforwardly results from expanding
  \eqref{eq2.13}.  To prove \eqref{eq2.15}, consider first the upper signs.  From the commutation relation \eqref{eq2.2}, one has
  \begin{equation} \label{eq2.16}
A_+ A_- A_+= -A_- A_+^2 + 2 \A_0 A_+ \, .
  \end{equation}  
  It hence follows that
  \begin{equation}
    \begin{split}
      P_i A_{+}A_{-} A_{+} P_j +  (i \leftrightarrow j)& = P_i (-A_- A_+^2 + 2 \A_0 A_+)P_j +   (i \leftrightarrow j) \\
      & =  -P_i A_- P_j A_+^2 + 2 \A_0 P_i  A_+ P_j +   (i \leftrightarrow j) \\
      & = -(P_i P_j A_- + A_- P_i P_j)A_+^2 + 2 \A_0 (P_i P_j A_+ + A_+ P_i P_j) \\
      & = P_i P_j (-A_- A_+^2 + 2 \A_0 A_+) + (-A_- A_+^2 + 2 \A_0 A_+) P_i P_j \\
      & =  P_i P_j A_{+}A_{-} A_{+}  +  A_{+ }A_{-} A_{+} P_i P_j,
    \end{split}  
  \end{equation}  
  where we have used $[P_i, A_+^2]= [P_i,\A_0]=0$, \eqref{eq2.14} and \eqref{eq2.16}.

  The proof of \eqref{eq2.15} with lower signs proceeds in the same way.
  \end{proof}  

\section{Centralizing elements and a generalized Bannai-Ito algebra}
\label{sec3}

We shall now introduce elements involving the supplementary involutions that will form a centralizer of $\osp(1,2)$.  This will be the extended algebra we are looking for and its defining relations will be given below.

Denote by $[n]= \{1,\dots, n \}$ the set of the first $n$ positive integers. Let $S=\{ s_1,\dots, s_k \}$  be an ordered $k$-subset of $[n]$. Write
\begin{equation}
P_S= P_{s_1} P_{s_2} \cdots P_{s_k}\, .
\end{equation}  
\begin{proposition} \label{prop3.1} The elements 
  \begin{align} \label{eq3.2a}
    \tag{3.2a}
       C_S& = \frac{1}{4} \bigl\{ A_-, [A_+, P_S] \bigr\} - \frac{1}{2}P_S \\
\tag{3.2b} \label{eq3.2b}
       & =  \frac{1}{4} \bigl\{ [P_S, A_-], A_+ \bigr\} - \frac{1}{2}P_S
  \end{align}
  \setcounter{equation}{2}
  associated to the set $S$,  centralize $\osp(1,2)$, that is, they satisfy
  \begin{equation}
[C_S, \A_0] = [C_S, A_{\pm}]= [C_S,P]=0\, .
  \end{equation}
  \end{proposition}  
\begin{proof}
  The equality of the formulas \eqref{eq3.2a} and \eqref{eq3.2b} 
  for $C_S$ is immediately obtained by expanding both expressions and using
  $\{ A_+, A_- \}=2 \A_0$ and $[P_S,\A_0]=0$.  The relations $[C_S,\A_0]=0$
  and $[C_S,P]=0$ follow from the fact that $C_S$ is bilinear in $A_+$
  and $A_-$, that $\A_0$ commutes with $P_S$ and that $\{ P, A_{\pm}\}=0$.
  Let us now show that $[C_S, A_-]=0$. First, one observes that 
  \begin{equation}
    \begin{split}
      [C_S, A_-] & =  \left[\frac{1}{4} A_- [A_+,P_S]+ \frac{1}{4}[A_+,P_S]A_--\frac{1}{2}P_S, A_-  \right] \\
        & =  \frac{1}{4} \bigl[ [A_+, P_S], A_-^2  \bigr] -\frac{1}{2}[P_S,A_-]\, .
      \end{split}
  \end{equation}  
  Using the Jacobi identity and the fact that the involutions commute with the $su(1,1)$ generators $B_{\pm}= A_{\pm}^2$, one completes the proof by observing that
  \begin{equation}
[C_S, A_-]=  -\frac{1}{4} \bigl[ [A_-^2, A_+ ], P_S  \bigr] -\frac{1}{2}[P_S,A_-]=0
  \end{equation}  
with the help of \eqref{eq2.5}. The demonstration that $[C_S, A_+]=0$ proceeds in exactly the same way when the expression \eqref{eq3.2b} for $C_S$ is used.
\end{proof}
\begin{remark} Owing to the fact that the involutions commute, $C_S$ is symmetric under the permutations of the elements of $S$. For example, in the case of subsets $S$ of cardinality $2$, the $2$-index $C_{ij}$, $i,j=1,\dots,n$, satisfy
  \begin{equation}
   C_{ij}= C_{ji},\quad i\neq j\, .
   \end{equation} 
\end{remark}  
\begin{remark}
  For $S=[n]$, $P_{[n]}=P$ by hypothesis and since $P A_{\pm}= -A_{\pm} P$,
  \begin{equation}
C_{[n]} = \frac{1}{4} \bigl\{ A_-, [A_+,P] \bigr\} -\frac{1}{2} P = \frac{1}{2} \bigl([A_-,A_+]-1 \bigr) P = \Gamma\, .
  \end{equation}  
\end{remark}  
The following corollary is immediate given that $[C_S,A_{\pm}]=0$ and $[C_S,P]=0$ by Proposition~\ref{prop3.1}.
\begin{corollary}\label{cor3.4}
  The elements $C_S$ defined in Proposition~\ref{prop3.1} also have the property of commuting with $C_{[n]}=\Gamma$:
  \begin{equation} \label{eq3.7}
[C_S, C_{[n]}]=0 \, .
  \end{equation}  
\end{corollary}
The next two lemmas provide relations between the 1-index $C_i$, $i=1,\dots,n$, and the involutions.
\begin{lemma} \label{lem3.4}
  The 1-index elements
  \begin{align}
    \tag{3.8a} \label{eq3.8a}
    C_i & = \frac{1}{4} \bigl\{ A_-, [A_+,P_i] \bigr\} -\frac{1}{2} P_i  \\
     \tag{3.8b} \label{eq3.8b}
    &= \frac{1}{4} \bigl\{ [P_i, A_-], A_+ \bigr\} -\frac{1}{2} P_i
  \end{align}
\setcounter{equation}{8}  
satisfy
\begin{equation} \label{eq3.9}
C_i  P_j+ C_j P_i = P_i  C_j+ P_j C_i 
\end{equation}
for all $i \neq j$.
\end{lemma}  
\begin{proof} Take $i\neq j$.  Beginning with \eqref{eq3.8a} and using \eqref{eq2.13}, one finds
  \begin{equation} \label{eq3.10}
[C_i,P_j]= \frac{1}{4} \Bigl[A_- [A_+, P_i] + [A_+, P_i] A_- , P_j \Bigr] =\frac{1}{4} \Bigl( [A_-,P_j] [A_+,P_i]+ [A_+,P_i] [A_-,P_j]  \Bigr) \, .
  \end{equation}  
  Similarly, starting with \eqref{eq3.8b}, one obtains
\begin{equation} \label{eq3.11}
[C_j,P_i]= \frac{1}{4} \Bigl[[P_j, A_-]A_+ + A_+ [P_j, A_-]  , P_i \Bigr] =\frac{1}{4} \Bigl( [P_j, A_-] [A_+,P_i]+ [A_+,P_i] [P_j, A_-]  \Bigr) \, .
  \end{equation}    
One then sees that
\begin{equation} \label{eq3.12}
[C_i,P_j]=-[C_j,P_i],\quad i\neq j
\end{equation}  
from where the relation \eqref{eq3.9} follows.
\end{proof}  
In view of the expression \eqref{eq3.10} (or \eqref{eq3.11}) and given \eqref{eq2.13}, it follows that
\begin{equation}
\bigl[ P_i, [P_j,C_k] \bigr] =0, \quad \text{for }i,j,k \text{ distinct.}
\end{equation}  
Lemma~\ref{lem3.4} admits a multi-index generalization.
\begin{lemma} \label{lem3.5}
  Let $S=\{ s_1,\dots, s_k \}$ be an ordered $k$-subset of $[n]$ and write
  \begin{equation}
P_S=\prod_{\ell=1}^k P_{s_\ell} \, .
  \end{equation}  
  One has
  \begin{equation} \label{eq3.15}
P_S \left(\sum_{i=1}^k P_{s_i} C_{s_i}  \right) = \left(\sum_{i=1}^k C_{s_i} P_{s_i}  \right) P_S\, .
  \end{equation}  
\end{lemma}  
\begin{proof} The above relation can be written in the form
  \begin{equation} \label{eq3.16}
  \sum_{i=1}^k [ P_{S\setminus \{ s_i\}}, C_{s_i}] =0\, .
  \end{equation}  
  Use induction.  We know that \eqref{eq3.15} is true for $k=2$ from Lemma~\ref{lem3.4}.  Assume that \eqref{eq3.15}, or equivalently \eqref{eq3.16}, is valid if the cardinality of $S$ is $k-1$.  This means that for some $j$ between 1 and $k$:
   \begin{equation} \label{eq3.17}
  \sum_{i=1; i \neq j}^k [ P_{S\setminus \{ s_i, s_j\}}, C_{s_i}] =0\, .
   \end{equation}
Note that
\begin{equation} \label{eq3.18}
P_{S\setminus \{ s_i\}} = P_{S\setminus \{ s_i, s_j\}} P_{s_j}, \quad j\neq i \, .
\end{equation}  
Use \eqref{eq3.18} to write
\begin{equation} \label{eq3.19}
  \sum_{i=1}^k [ P_{S\setminus \{ s_i\}}, C_{s_i}]
  = \frac{1}{k}  \sum_{\ell=1}^k  \left(  [ P_{S\setminus \{ s_\ell\}}, C_{s_\ell}] + 
   \sum_{j=1;\, j\neq \ell}^k [ P_{S\setminus \{ s_\ell,s_j\}} P_{s_j}, C_{s_\ell}]  \right) \, .
\end{equation}  
Some commutator algebra thus gives
\begin{equation} \label{eq3.20}
\left( \frac{k-1}{k} \right)  \sum_{i=1}^k [ P_{S\setminus \{ s_i\}}, C_{s_i}]
  = \frac{1}{k}   \sum_{\substack{j,\ell=1\\ j\neq \ell}}^k \Bigl(  [ P_{S\setminus \{ s_\ell,s_j\}} , C_{s_\ell}] P_{s_j}+ 
     P_{S\setminus \{ s_\ell,s_j\} \\}  [P_{s_j}, C_{s_\ell}]  \Bigr) \, .
\end{equation}  
Since $[P_{s_j},C_{s_\ell}]=-[P_{s_\ell}, C_{s_j}]$ per \eqref{eq3.12} the terms $P_{S\setminus \{ s_\ell,s_j\}}  [P_{s_j}, C_{s_\ell}]$
cancel two by two in the sums and we are left with
\begin{equation}
\sum_{i=1}^k [ P_{S\setminus \{ s_i\}}, C_{s_i}]
  = \frac{1}{k-1}  \sum_{j=1}^k  \left(  
   \sum_{\ell=1;\,  \ell\neq j}^k [ P_{S\setminus \{ s_\ell,s_j\}} , C_{s_\ell}]  \right) P_{s_j}
\end{equation}
which is zero according to the induction hypothesis \eqref{eq3.17}.
\end{proof}  

Proposition~\ref{prop3.1} shows that when the supplementary involutions $P_i$, $i=1,\dots,n$, are present, a centralizer of $\osp(1,2)$ is generated by the elements $C_S$ corresponding to all subsets of $[n]$ with $k$ ordered elements for $k=1,\dots,n-1$.
Since we wish to extend the rank-1 Bannai-Ito algebra, we shall take $n=3$.  A centralizer of $\osp(1,2)$ will then be generated by the 1 and 2-index elements $C_1$, $C_2$, $C_3$ and  $C_{12}$, $C_{23}$, $C_{13}$ to which the Casimir element $C_{123}$ can be added.  Remember that $C_{123}$
commutes with $C_i$, $C_{ij}$, $i,j=1,2,3$, but
does not commute with the involutions $P_i$, $i=1,2,3$.  It shall be seen that the elements $C_i$, $i=1,2,3$, together with $C_{ijk}$ take the place of the structure constants and complement the elements $C_{ij}$, $ij=12,13,23$, that generalize the generators $K_{ij}$ of the Bannai-Ito algebra.  The task is now to obtain the relations that these elements obey.

It is useful to record the special form that some relations take when the indices are restricted to the set $\{1,2,3 \}$.  Take the indices $i,j,k\in \{ 1,2,3 \}$ to be all distinct and thus forming a permutation of $(1,2,3)$.
Clearly,
\begin{equation}
P=P_i P_j P_k, \quad i,j,k \in \{1,2,3\} \text{ all distinct.}
\end{equation}  
Given that $P A_{\pm}= - A_{\pm} P$, formulas of the type
\begin{equation} \label{eq3.23}
P_i P_j A_{\pm} P_k= -P_k A_{\pm} P_i P_j
\end{equation}  
follow immediately.  The result of Lemma~\ref{lem3.5} will specialize to
\begin{equation} \label{eq3.24}
[P_i P_j, C_k] + [P_j P_k, C_i] + [P_i P_k, C_j] =0 \, .
\end{equation}  
Moreover,
\begin{equation}
C_{ijk}=C_{123}=\Gamma \, .
\end{equation}
It is in fact possible to show that this Casimir operator can be written as a combination of the 1- and 2-index elements $C_i$ and $C_{ij}$ together with the involutions $P_i$.
\begin{proposition} \label{prop3.6} When $n=3$ and $i,j,k \in \{1,2,3\}$ are all
  distinct, the Casimir element $C_{ijk}$ has the following expression:
  \begin{equation} \label{eq3.26}
C_{ijk} = C_{ij} P_k + C_{jk} P_i + C_{ik} P_j -C_k P_i P_j -C_i P_j P_k-C_j P_i P_k-\frac{1}{2} P_i P_j P_k \, .
  \end{equation}
  \end{proposition}    
  \begin{proof}
    Using the definition \eqref{eq3.2a} and \eqref{eq3.23}, one sees that
    \begin{equation}
C_{ij} P_k-C_k P_iP_j= \frac{1}{2} \bigl( A_- P_k A_+ P_i P_j-  A_+ P_k A_- P_i P_j \bigr) \, .
    \end{equation}
    Now
\begin{equation}
\begin{split}
   \frac{1}{2} A_- P_k A_+ P_i P_j &+ \text{ cyclic permutations}= 
   \frac{1}{4} \bigl(  A_- P_k A_+ P_i P_j+  A_- P_i A_+ P_j P_k \bigr)\\
  & + \frac{1}{4} \bigl(  A_- P_i A_+ P_j P_k+  A_- P_j A_+ P_i P_k \bigr) + \frac{1}{4} \bigl(  A_- P_j A_+ P_i P_k+  A_- P_k A_+ P_i P_j \bigr) \, .
\end{split}
\end{equation}
With the help of \eqref{eq2.14} and of \eqref{eq3.23} again, one finds
\begin{equation}
  \begin{split}
     \frac{1}{2} A_- P_k A_+ P_i P_j &+ \text{ cyclic permutations}= \\
& \qquad    \frac{3}{4} A_- A_+ P- \left( \frac{1}{4}  A_- P_k A_+ P_i P_j
  + \text{ cyclic permutations}  \right) \, ,
\end{split}
\end{equation}
from where one obtains that
\begin{equation}
 A_- P_k A_+ P_i P_j + \text{ cyclic permutations} = A_- A_+ P \, .
\end{equation}  
Similarly one gets
\begin{equation}
 A_+ P_k A_- P_i P_j + \text{ cyclic permutations} = A_+ A_- P \,.
\end{equation}  
One thus sees that
\begin{equation}
\Bigl( C_{ij} P_k -C_k P_{ij} +  \text{ cyclic permutations}\Bigr) -\frac{1}{2}P_i P_j P_k = \frac{1}{2} \bigl(  A_- A_+- A_+ A_- -1\bigr)P = \Gamma 
\end{equation}  
which completes the proof.
  \end{proof}
We are now ready to state our main result that gives the commutation relations of the elements $C_i$ and $C_{ij}$, $i,j=1,2,3$ and to offer indications on how it is obtained.
\begin{theorem} \label{theo3.7}
  For $i,j,k \in \{ 1,2,3\}$ and all distinct, the elements $C_i$ and $C_{ij}$ satisfy the following relations
  \begin{equation} \label{eq3.33}
  \{C_{ij}, C_{jk}  \} = C_{ik} +  \{C_{j}, C_{ijk}  \} + \{C_{i}, C_{k}  \} 
  \end{equation}
  and
  \begin{equation} \label{eq3.34}
  [C_{ij}, C_{k}  ] +  [C_{jk}, C_{i} ] + [C_{ik}, C_{j}]=0 \, .
  \end{equation}
\end{theorem}  
Although a bit involved, the proofs proceed straightforwardly.  The basic strategy is to separately expand and simplify the expressions on the left and right hand sides to show that they coincide.  From the definitions of $C_i$, $C_{ij}$
and $C_{ijk}$, we see that \eqref{eq3.33} and \eqref{eq3.34} will amount to equalities between expressions that involve terms that are quartic, quadratic and of degree 0 in the $\osp(1,2)$ generators $A_-$ and $A_+$ which always occur in pairs (so as to commute with $A_0$).  Expand those expressions in full. It is practical to focus on each of these categories of terms and to verify that they are identically equal on their own upon comparing left- and right-hand sides.  Furthermore, when dealing with the quartic component, it is also possible to concentrate in turn on terms of definite type with respect to the ordering, e.g. terms having the structure $A_+ A_-^2 A_+$, $A_- A_+^2 A_-$ or $A_+ A_- A_+ A_-$
and $A_- A_+ A_- A_+$ together, for instance.  In order to not overly clutter the flow of the presentation, we shall relegate further discussion of the proofs to the Appendix, where examples of the computations that are needed will be given.

\begin{remark} \label{rem3.8}  In view of Corollary~\ref{cor3.4}, we have in addition
  \begin{equation}
[C_i, C_{ijk}]=[C_{ij}, C_{ijk}]=0\, , \quad i,j,k \text{ distinct.}
  \end{equation}  
  It follows that the term $\{ C_j, C_{ijk}\}$ in \eqref{eq3.33} can also be written as $2\, C_j C_{ijk}$.
  The relations \eqref{eq3.33} and \eqref{eq3.34} thus bear a kinship with those of the Bannai-Ito algebra.  In a situation where the $C_i$ are constants, \eqref{eq3.34} trivializes and one recovers the relations \eqref{eq1.1}
with $\omega_{ij}$ elements of the center of the algebra generated by $C_{ij}$. We shall discuss in the next section an instance where this happens.
\end{remark} 

\section{Embeddings of $\osp(1,2)$ into $\osp(1,2)^{\otimes 3}$} \label{sec4}

A situation in which three involutions naturally arise is when one considers the coproduct embedding of $\osp(1,2)$ into its three-fold tensor product 
$\osp(1,2)^{\otimes 3}$.  We shall now indicate how this fits within the framework that has been developed and how previously known results appear. Let us recall the definition of the $\osp(1,2)$ coproduct:  it is a co-associative homomorphism
\begin{equation}
\Delta: \osp(1,2) \to \osp(1,2) \otimes \osp(1,2)
\end{equation}  
such that
\begin{equation}
  \begin{split}
    \Delta(A_{\pm}) & =A_{\pm} \otimes P + 1 \otimes A_{\pm} \\
    \Delta(\A_{0}) & = \A_{0} \otimes 1 + 1 \otimes \A_{0} \\
    \Delta (P) & = P \otimes P \, .
    \end{split}
\end{equation} 
Upon letting
\begin{equation}
\Delta^{(2)} = (\Delta \otimes 1) \circ \Delta
\end{equation}  
one has a homomorphism $\Delta^{(2)}: \osp(1,2) \to \osp(1,2)^{\otimes 3}$.
Introduce the three involutions
\begin{equation}
 P^{(1)}= P \otimes 1 \otimes 1,\qquad
P^{(2)}= 1 \otimes P \otimes 1,\qquad P^{(3)}= 1 \otimes 1 \otimes P \, .
\end{equation}
We have
\begin{equation} \label{eq4.5}
  \begin{split}
    \mathcal \A_0 &  = \Delta^{(2)}(\A_0)= \A_0^{(1)} + \A_0^{(2)}+\A_0^{(3)} \\
    \mathcal A_{\pm} &  = \Delta^{(2)}(A_{\pm})= A_{\pm}^{(1)} P^{(2)}P^{(3)}
    +  A_{\pm}^{(2)} P^{(3)} +  A_{\pm}^{(3)}  \\
  \mathcal P & = \Delta^{(2)}(P)=   P^{(1)} P^{(2)}P^{(3)}   
  \end{split}  
\end{equation}  
where the suffix designates to which factor the element belongs.
Since $\Delta^{(2)}$ is a homomorphism, we have a version of $\osp(1,2)$
generated by $ \mathcal \A_0$, $ \mathcal A_{\pm}$ and $\mathcal P$
with three involutions $P^{(i)}$, $i=1,2,3$, that satisfy the requirements of our set-up, namely
\begin{equation}
P^{(1)}  P^{(2)}  P^{(3)} = \mathcal P
\end{equation}
and
\begin{equation} \label{eq4.7}
[P^{(i)},\mathcal A_0]=0,\qquad \Bigl[ P^{(i)}, [P^{(j)}, \mathcal A_{\pm}]\Bigr]=0,\qquad i,j=1,2,3, \quad  i\neq j\, .
\end{equation}
The second relation of \eqref{eq4.7} is verified since
\begin{equation} \label{eq4.8}
\begin{split}  
  [\mathcal A_{\pm}, P^{(1)}] & = 2 A_{\pm}^{(1)} P^{(1)}  P^{(2)}  P^{(3)} \\
  [\mathcal A_{\pm}, P^{(2)}] & = 2 A_{\pm}^{(2)}  P^{(2)}  P^{(3)} \\
  [\mathcal A_{\pm}, P^{(3)}] & = 2 A_{\pm}^{(3)} P^{(3)} \, .
\end{split}
\end{equation}
Let us then examine what our construction entails in this case.  Consider first the 1-index elements $C_i$ and take $i=1$ for instance
\begin{equation}
    C_1  = \frac{1}{4} \Bigl\{  \mathcal A_-,  [\mathcal A_+, P^{(1)}]\Bigr\}- \frac{1}{2} P^{(1)} 
    = \frac{1}{2} [\mathcal A_-,A^{(1)}_+] P - \frac{1}{2} P^{(1)} = \frac{1}{2} \Bigl( [A_-^{(1)},A^{(1)}_+] - 1 \Bigr)P^{(1)}= \Gamma^{(1)}
\end{equation}  
where $\Gamma^{(1)}$ stands for the Casimir operator of the first factor in
$\osp(1,2)^{\otimes 3}$.  Similarly, we find that in fact
\begin{equation}
C_i = \Gamma^{(i)},\qquad i=1,2,3\, .
\end{equation}
The various $C_i$ manifestly all commute.
In this case, one refers to the simple embeddings of $\osp(1,2)$ into $\osp(1,2)^{\otimes 3}$ as one of the factors.  In addition to those and to $\Delta^{(2)}$ there are other embeddings labelled by two indices for which the generators are:
\begin{align} \label{eq4.11}
 & \mathcal A_0^{(ij)} = \A_0^{(i)}+\A_0^{(j)}, \qquad  i,j=1,2,3   \nonumber \\
  \mathcal A_{\pm}^{(12)} = A_{\pm}^{(1)} P^{(2)}+A_{\pm}^{(2)}, \quad
 &  \mathcal A_{\pm}^{(23)} = A_{\pm}^{(2)} P^{(3)}+A_{\pm}^{(3)}, \quad
  \mathcal A_{\pm}^{(13)} = A_{\pm}^{(1)}P^{(2)} P^{(3)}+A_{\pm}^{(3)} \nonumber \\
   \mathcal P^{(12)}= P^{(1)} P^{(2)}, & \qquad  
    \mathcal P^{(23)}= P^{(2)} P^{(3)}, \qquad
    \mathcal P^{(13)}= P^{(1)} P^{(3)} \, .
\end{align}  
To these embeddings correspond the Casimirs
\begin{equation}
\Gamma^{(ij)}=\frac{1}{2} \Bigl( [\mathcal A_{-}^{(ij)}, \mathcal A_{+}^{(ij)}]-1 \Bigr) \mathcal P^{(ij)} \, .
\end{equation}
It is immediate to observe that the Casimir operators $C_i = \Gamma^{(i)}$ of each of the factors commute with all the generators in \eqref{eq4.11}
and hence with the intermediate Casimirs, namely
\begin{equation}
[\Gamma^{(i)}, \Gamma^{(jk)}]=0\,, \quad \forall \, i,j,k \, \, (j\neq k)\, .
\end{equation}  
Let us now come to the 2-index elements and consider for instance
\begin{equation}
C_{13}= \frac{1}{4} \Bigl \{ \mathcal A_-, [\mathcal A_+, P^{(1)} P^{(3)}] \Bigr \}-\frac{1}{2}\mathcal P^{(13)}\, .
\end{equation}  
With the help of \eqref{eq4.8}, we see that
\begin{equation}
  \begin{split}
    C_{13}& =  \frac{1}{4} \Bigl \{ \mathcal A_-, [\mathcal A_+, P^{(1)} ]P^{(3)}
    + P^{(1)} [\mathcal A_+,  P^{(3)}] \Bigr \}-\frac{1}{2}\mathcal P^{(13)} \\
    & =  \frac{1}{2} \Bigl \{ \mathcal A_-, A_+^{(1)} P^{(1)} P^{(2)}
    +  A_+^{(3)} P^{(1)} P^{(3)}  \Bigr \}-\frac{1}{2}\mathcal P^{(13)} \\
     & =  \frac{1}{2} \Bigl \{ \mathcal A_-,  \mathcal A_+^{(13)} P^{(1)} P^{(3)}  \Bigr \}-\frac{1}{2}\mathcal P^{(13)}\, .
  \end{split}  
\end{equation}  
Now $\mathcal A_-=\mathcal A_-^{(13)}+ A_-^{(2)} P^{(3)}$ and
\begin{equation}
  \{ A_-^{(2)} P^{(3)},  \mathcal A_+^{(13)} P^{(1)} P^{(3)} \}
  =  \{ A_-^{(2)} P^{(3)},  A_+^{(1)} P^{(1)} P^{(2)} +  A_+^{(3)} P^{(1)} P^{(3)} \}=0\, .
\end{equation}  
Hence
\begin{equation}
C_{13}=\frac{1}{2} \Bigl( [\mathcal A_{-}^{(13)}, \mathcal A_{+}^{(13)}]-1 \Bigr) \mathcal P^{(13)}= \Gamma^{(13)}
\end{equation}  
and we observe that $C_{13}$ is the intermediate Casimir element associated to the embedding $(ij)=(13)$ of \eqref{eq4.11}.  In general, we find that
\begin{equation}
C_{ij}= \Gamma^{(ij)} \, .
\end{equation}  
As always,
\begin{equation}
C_{ijk}= \frac{1}{2}  \Bigl( [\mathcal A_{-}, \mathcal A_{+}]-1 \Bigr) \mathcal P 
\end{equation}  
is the Casimir element of the main $\osp(1,2)$ of the construction, here
$\Delta^{(2)}\bigl( \osp(1,2) \bigr)$.

The centralizer of $\Delta^{(2)}\bigl( \osp(1,2) \bigr)$
is therefore the algebra generated by the various Casimir operators: $C_i=\Gamma^{(i)}$,  $C_{ij}=\Gamma^{(ij)}$ and $C_{ijk}=\Gamma$.
Equation \eqref{eq3.34} is hence trivially satisfied and $\{ C_j, C_{ijk} \}=2 \Gamma^{(j)} \Gamma$ while  
$\{ C_i, C_{k} \}=2 \Gamma^{(i)} \Gamma^{(k)}$.  It thus follows that the relations of Theorem~\ref{theo3.7} reduce to those of the Bannai-Ito algebra with $\omega_{ij}$ central in this embedding situation.  Furthermore, if we were to consider products of three irreducible representations, the Casimirs $\Gamma^{(i)}$ would be constants.  This confirms that the defining relations \eqref{eq3.33}-\eqref{eq3.34} generalize those of the Bannai-Ito algebra whose appearance in the Racah problem for $\osp(1,2)$ is here seen as a special case in our framework.

\section{An hyperoctahedral extension of the Bannai-Ito algebra}
\label{sec5}

We shall obtain in this section a generalization of the Bannai-Ito algebra that involves the signed symmetric group in three objects.  This will be achieved by considering a generalization of $\osp(1,2)$ built from Dunkl operators \cite{cit12}.  Let $x_i$, $i=1,2,\dots,n$,
be the variables; we shall keep their number arbitrary for the moment.  We wish to realize the involutions $P_i$ by the reflections $R_i$ in the plane $x_i=0$, i.e.
\begin{equation}
R_i f(\dots,x_i,\dots)=  f(\dots,-x_i,\dots) \, .
\end{equation}  
It is hence natural to consider the Dunkl operators associated with the Weyl group of type $B_n$ which contains the reflections $R_i$ and the permutations $\pi_{ij}$:
\begin{equation}
\pi_{ij} f(\dots, x_i,\dots,x_j,\dots)=  f(\dots, x_j,\dots,x_i,\dots)\, .
\end{equation}  
These $B_n$-Dunkl operators depend on two parameters $a$ and $b$, and are defined by \cite{cit13,cit14} 
\begin{equation} \label{eq5.3}
D_i = \frac{\partial}{\partial x_i} + \frac{b}{x_i} (1-R_i) + a \sum_{j \neq i} \left( \frac{1}{x_i-x_j}(1-\pi_{ij})+  \frac{1}{x_i+x_j}(1-R_i R_j \pi_{ij})  \right) \, .
\end{equation}
Remarkably, they form a commuting set:
\begin{equation} \label{eq5.4}
[D_i,D_j]=0 \qquad \forall \, i,j \, .
\end{equation}  
Their commutation relations with the coordinates are found to be
\begin{equation} \label{eq5.5}
[D_i,x_j]= \delta_{ij} \left( 1+a \sum_{k \neq i} (1+R_i R_k) \pi_{ik} + 2b R_i  \right) -(1-\delta_{ij}) a (1-R_i R_j) \pi_{ij} \, .
\end{equation}  
A special feature of the Dunkl operators of type $B_n$ is their behavior under
reflections.  Indeed, using $R_i \pi_{ij}= \pi_{ij} R_j$, it is easy to see that
\begin{equation}\label{eq5.6}
\{R_i, D_i\}=0, \qquad [R_i,D_j]=0\quad i\neq j \, .
\end{equation}
These relations are also valid of course when $D_i$ is replaced by the coordinate $x_i$.
Property \eqref{eq5.6} is key.  Let us now restrict to the situation with three variables $x_1$, $x_2$, $x_3$ and the corresponding $B_3$-Dunkl operators $D_1$, $D_2$, $D_3$.
Take
\begin{equation} \label{eq5.7}
  \begin{split}
\hat A_-& = D_1 R_2 R_3 + D_2 R_3 + D_3 \\
\hat A_+& = x_1 R_2 R_3 + x_2 R_3 + x_3 \, .
  \end{split}  
\end{equation}  
In view of the reflection properties \eqref{eq5.6} and given \eqref{eq5.4},
one sees that
\begin{equation}
  \begin{split}
\hat B_-& = A_-^2= D_1^2 + D_2^2 + D_3^2 \\
\hat B_+& = A_+^2 =x_1^2 + x_2^2 + x_3^2 \, .
  \end{split}  
\end{equation}  
It is known \cite{cit15,cit16} that such operators $\hat B_{\pm}$ provide a realization of the $su(1,1)$ commutation relations \eqref{eq2.7} with the Euler operator
\begin{equation} \label{eq5.9}
  \hat A_0= x_1 \frac{\partial}{\partial x_1} + x_2 \frac{\partial}{\partial x_2}+
  x_3 \frac{\partial}{\partial x_3} + 3 \left( 2a+b + \frac{1}{2} \right)
\end{equation}  
playing the role of $\A_0$.  It is in fact not too difficult 
to see moreover that $\hat A_{\pm}$ close onto the $\osp(1,2)$ relations:
\begin{equation}
\{ \hat A_-, \hat A_+ \}= 2 \hat A_0, \qquad [\hat A_0, \hat A_{\pm}]=\pm \hat A_{\pm} \, .
\end{equation}  
Let us remark that the form of $\hat A_-$ and $\hat A_+$ in \eqref{eq5.7}
is reminiscent of the expressions we had for $\mathcal A_-$ and $\mathcal A_+$
in Section~\ref{sec4} as a result of the coproduct action.  The key difference though is that we no longer have the analogy of $[A_-^{(i)},A_+^{(j)}]=0$, $i \neq j$, as is seen from \eqref{eq5.5}. 

The reader might also think that the Dunkl operators
\begin{equation}
\nabla_i = \frac{\partial}{\partial x_i} + \frac{\mu_i}{x_i} (1-R_i),\quad i=1,2,3 
\end{equation}  
associated to the reflection group $\mathbb Z_2^3$ and with $\mu_1$, $\mu_2$, $\mu_3$ constants, could prove pertinent to introduce the reflections $R_i$ as involutions.  This is true and has in fact been considered already \cite{cit7,cit8}; note however that
\begin{equation}
[\nabla_i,x_j]=0 \quad i\neq j
\end{equation}  
in this case which leads to a realization of the embedding situation of Section~\ref{sec4}.  The uniform specialization $\mu_1=\mu_2=\mu_3=b$ arises with the $B_3$-operators when $a=0$.

Before going further we should confirm that the requirements on the involutions are satisfied.  This is so because $P$ is realized as $R=R_1 R_2 R_3$; indeed
\begin{equation}
[R, \hat A_0]=0,\qquad \{ R, \hat A_{\pm}\}= 0 \, .
\end{equation}  
Moreover, owing to \eqref{eq5.7} we see that
\begin{equation}
\Bigl[ R_i, [R_j,  \hat A_{\pm}] \Bigr]=0 \qquad i\neq j \, .
\end{equation}  
We can thus apply our formalism to get the explicit expressions of the centralizing elements in this concrete realization of $\osp(1,2)$ and to determine as well the defining relations of the particular generalized Bannai-Ito algebra that emerges.

Let
\begin{equation}
S_{ij}= [D_i, x_j] \, .
\end{equation}

\noindent We see from \eqref{eq5.5} that $S_{ij}=S_{ji}$, namely that
$[D_i, x_j]=[D_j, x_i]$.  It is useful to record for reference the
explicit expressions for $S_{ij}$ for $i\leq j$, $i,j=1,2,3$:
\begin{equation} \label{eq5.16}
  \begin{split}
    S_{11}& = 1+ a(1+ R_1 R_2)\pi_{12} + a(1+ R_1 R_3)\pi_{13}+2b R_1 \\
    S_{22}& = 1+ a(1+ R_1 R_2)\pi_{12} + a(1+ R_2 R_3)\pi_{23}+2b R_2 \\
    S_{33}& = 1+ a(1+ R_1 R_3)\pi_{13} + a(1+ R_2 R_3)\pi_{23}+2b R_3 \\
    S_{12}& = -a(1-R_1 R_2) \pi_{12} \\
    S_{13}& = -a(1-R_1 R_3) \pi_{13} \\
     S_{23}& = -a(1-R_2 R_3) \pi_{23} \, .
  \end{split}  
\end{equation}  
Let
\begin{equation}
M_{ij}= x_i D_j -x_j D_i , \qquad i,j=1,2,3 \, .
\end{equation}  
The commutation relation of these Dunkl angular momentum operators have been given in \cite{cit17} and read
\begin{equation}
  \begin{split}
    [M_{ij}, M_{k\ell}]& = M_{i \ell} S_{jk} + M_{jk} S_{i\ell}-M_{i k} S_{\ell j}-M_{j \ell} S_{ik}\\
    & = S_{jk} M_{i \ell}  + S_{i \ell} M_{jk} -S_{\ell j} M_{i k} -S_{ik} M_{j \ell} \, .
    \end{split}
\end{equation}  
Using $C_{ij}= \bigl\{\hat A_-,[\hat A_+, R_i R_j] \bigr\}/4 -R_i R_j/2$, it is straightforward to arrive at the following result.
\begin{proposition}
  The 2-index elements that commute with $\hat A_-$ and $\hat A_+$ are
  \begin{equation} \label{eq5.19}
    \begin{split} 
      C_{12}& = M_{12} R_1 + \frac{1}{2} (S_{11}+ S_{22}-1)R_1 R_2 -\frac{1}{2}(S_{13}+S_{23} R_1) \\
      C_{23}& = M_{23} R_2 + \frac{1}{2} (S_{22}+ S_{33}-1)R_2 R_3 -\frac{1}{2}(S_{12}R_3+S_{13}) \\
       C_{13}& = M_{13} R_1R_2 + \frac{1}{2} (S_{11}+ S_{33}-1)R_1 R_3 -\frac{1}{2}(S_{12}R_3+S_{23} R_1) \, .
    \end{split}    
  \end{equation}
\end{proposition}  
We note that in obtaining the final form of the formulas \eqref{eq5.19} one uses
\begin{equation} \label{eq5.20}
S_{ij}R_i R_j = -S_{ij}\, \qquad i \neq j\, .
\end{equation}

Introduce now the quantities $Q_{ij}=Q_{ji}$ in the group algebra of $B_3$ that are defined as follows:
\begin{equation}
  \begin{split}
    Q_{12}& = \frac{1}{2}(1+R_1+R_2-R_1R_2) \pi_{12} \\
    Q_{13}& = \frac{1}{2}(R_1+R_2+R_3-R_1R_2R_3) \pi_{13} \\
     Q_{23}& = \frac{1}{2}(1+R_2+R_3-R_2R_3) \pi_{23} \, .
   \end{split}  
\end{equation}  
One first observes that these $Q_{ij}$ are involutions:
\begin{equation} \label{eq5.22}
Q_{ij}^2=1 \, .
\end{equation}  
Furthermore, one finds that the algebra they form is isomorphic to the algebra of the permutation operators
$\pi_{12}$, $\pi_{13}$ and $\pi_{23}$ of the symmetric group $S_3$.  Indeed, one checks that 
\begin{equation} \label{eq5.23}
Q_{12} Q_{13}= Q_{23}Q_{12}= Q_{13}Q_{23}, \qquad Q_{12} Q_{23}= Q_{23}Q_{13}= Q_{13}Q_{12}
  \end{equation}  
together with $Q_{ij} R_j=R_i Q_{ij}$.  One discovers also that
\begin{equation}
[Q_{ij}, C_{ij}]=0
\end{equation}  
and that
\begin{equation} \label{eq5.25}
  \begin{split}
    Q_{12} C_{13} & = C_{23}Q_{12},\qquad  Q_{12} C_{23}= C_{13}Q_{12}, \qquad  Q_{13} C_{12}= C_{23}Q_{13}  \\
     Q_{13} C_{23} & = C_{12}Q_{13},\qquad  Q_{23} C_{12}= C_{13}Q_{23}, \qquad  Q_{23} C_{13}= C_{12}Q_{23}  \, .
  \end{split}  
\end{equation}  
Consider now the 1-index elements  $C_{i}= \bigl\{\hat A_-,[\hat A_+, R_i] \bigr\}/4 -R_i /2$.  One readily finds
\begin{equation} \label{eq5.26}
  \begin{split}
    C_1 & = \frac{1}{2} (S_{11}R_1+S_{12}R_1 R_2+S_{13}R_1 R_2 R_3 -R_1) \\
    C_2 & = \frac{1}{2} (-S_{12}+S_{22}R_2-S_{23}-R_2) \\
     C_3 & = \frac{1}{2} (-S_{13}R_2-S_{23}+S_{33}R_3 -R_3) \, .
  \end{split}  
\end{equation}  
These expressions can be rewritten in terms of the involutions $Q_{ij}$, $i,j=1,2,3$.  The results read as follows.
\begin{proposition}
  The centralizing 1-index elements $C_i$, $i=1,2,3$, are given by
  \begin{equation} \label{eq5.27}
    C_i = a(Q_{ij}+ Q_{ik}) +b
  \end{equation}
  with $i,j,k \in \{1,2,3 \}$ and all distinct.
\end{proposition}  
\begin{remark}
  Since
  \begin{equation}
    C_i + C_j -C_k = 2aQ_{ij} +b
  \end{equation}
  it follows that the quantities $Q_{ij}$ commute also with $\hat A_{\pm}$.  We thus see that these involutions enlarge the
  set of symmetries formed by the operators $C_{ij}$.
\end{remark}  
\begin{remark}
  One can now verify that 
  \begin{equation} \label{eq5.29}
     C_{ij}= M_{ij} R_i R_{j-1} +C_i R_j+C_j R_i +\frac{1}{2} R_i R_j 
  \end{equation}
  where $R_i R_{j-1}=R_i$ if $i=j-1$. 
\end{remark}  
\begin{proposition}
  In the realization \eqref{eq5.7}, \eqref{eq5.9} of $\osp(1,2)$ with the $B_3$-Dunkl operators, the Casimir element
  $\Gamma= \frac{1}{2}\bigl( [\hat A_-, \hat A_+]-1\bigr)R$, $R=R_1 R_2R_3$, is given by 
  \begin{equation} \label{eq5.30}
   \Gamma= M_{12} R_1 R_3 +  M_{13} R_1 +  M_{23} R_1 R_2 + \frac{1}{2} ( S_{11}+S_{22}+S_{33}-1)R\, .
  \end{equation}  
\end{proposition}  
\begin{proof} This is obtained by a direct computation. 
\end{proof}  
\begin{remark} Using \eqref{eq5.19} and \eqref{eq5.26} and invoking property \eqref{eq5.20} again, 
  one can see that \eqref{eq5.30} matches with formula \eqref{eq3.26}.  Given \eqref{eq3.26}, if one simplifies
  the part $C_i R_j R_k + C_j R_i R_k +C_k R_i R_j= \frac{1}{2} (S_{11}+S_{22}+S_{33}-3)R$, one may also write $\Gamma$ in the form
  \begin{equation} \label{eq5.31}
    \begin{split}
      \Gamma & = C_{12}R_3+C_{13}R_2+C_{23}R_1 \\
      & \quad - \left( a \Bigl(  (1+R_1 R_2)\pi_{12}+ (1+R_1 R_3)\pi_{13}+ (1+R_2 R_3)\pi_{23} \Bigr) + b(R_1+R_2+R_3) +\frac{1}{2} \right) R\, .
\end{split}
    \end{equation}  
\end{remark}
The Jucys-Murphy elements \cite{cit18,cit19} of $B_n$ have been given in \cite{cit20}.  For $B_3$ they are
\begin{equation}
    R_1, \quad R_2, \quad R_3, \quad m_2= (1+R_1R_2)\pi_{12} , \quad m_3= (1+R_1R_3)\pi_{13} + (1+R_2R_3)\pi_{23}\, .
\end{equation}  
It can be checked that all these elements commute with each other.  The symmetric polynomials in these entities are known to generate
the center of the group algebra of $B_3$.  It is thus not surprising the above Jucys-Murphy elements appear in the expression \eqref{eq5.31}
of the central Casimir operator $\Gamma$.

We are now ready to determine the specific form that the commutation relations \eqref{eq3.33} and \eqref{eq3.34} take.
\begin{proposition} Let $i,j,k \in \{ 1,2,3\}$ be distinct.  The defining relation \eqref{eq3.33} of the algebra generated by
  the operators $C_i$, $C_{ij}$ and $C_{ijk}$ respectively given by \eqref{eq5.27}, \eqref{eq5.29} and \eqref{eq5.31} for instance,
  takes the form
  \begin{equation} \label{eq5.33}
  \{ C_{ij}, C_{jk}\}= C_{ik}+ 2 \Gamma \Bigl( a(Q_{ij}+Q_{jk})+b \Bigr) +a^2\bigl( 3 \{ Q_{ij}, Q_{jk}\} +2 \bigr) +2ab (Q_{ij}+Q_{jk}+2Q_{ik}) + 2b^2\, .
  \end{equation}  
\end{proposition}  
\begin{proof} This simply follows from \eqref{eq3.33} and the expression of the 1-index elements $C_i$. Given \eqref{eq5.31}, one may directly verify  that
  $\{ \Gamma,C_j \}=2 \Gamma C_j$.  One shall note that
  \begin{equation}
\{Q_{ij}+ Q_{ik} , Q_{ik}+Q_{jk} \}= 3 \{ Q_{ij}, Q_{jk} \} +2
  \end{equation}  
  a relation easily derived from \eqref{eq5.23}.
\end{proof}  
\begin{remark} The relation \eqref{eq3.34} is readily seen to be identically satisfied in the present realization.  Indeed, we have
  \begin{equation}
 [C_{12},C_3]+ [C_{23},C_1]+ [C_{13},C_2]=  [C_{12},Q_{13}+Q_{23}] +  [C_{23},Q_{12}+Q_{13}]+  [C_{13},Q_{12}+Q_{23}]
  \end{equation}
  and one shows that these terms add up to zero with the help of relations \eqref{eq5.25}.
\end{remark}  
\begin{remark} We observe that when $a=0$, the relations \eqref{eq5.33} correspond to the relations \eqref{eq1.1} of the Bannai-Ito
  algebra with central structure constants given by $\omega_{12}=\omega_{13}=\omega_{23}=2 \Gamma b + 2b^2$. When $a\neq 0$, we have a remarkable generalization that involves the generators $\pi_{ij}$ and $R_i$, $i,j \in \{ 1,2,3\}$, of the signed permutation group on three objects.
 This hyperoctahedral extension of the Bannai-Ito algebra has $C_{ij}$ and $Q_{ij}$ as generators with relations given by \eqref{eq5.22}-\eqref{eq5.25} and \eqref{eq5.33} where $C_{ijk}$ is central. 
\end{remark}  
   
\begin{proposition} The generalized Bannai-Ito algebra defined by \eqref{eq5.33} admits a Casimir operator $C$ such that
  \begin{equation}
   [C,C_{ij}]= [C,C_i]=0
  \end{equation}  
  which is given by
  \begin{equation}
   C= C_{12}^2 +  C_{13}^2 +  C_{23}^2-a^2 Q^2 - 4ab Q
  \end{equation}  
  with
  \begin{equation}
    Q= Q_{12} +  Q_{13} +  Q_{23} \, .
  \end{equation}  
  \end{proposition}
\begin{proof} A direct calculation shows on the one hand that
  \begin{equation}
   [ C_{12}^2 +  C_{13}^2 +  C_{23}^2, C_{12}]= 3a^2 \bigl[ \{ Q_{12},Q_{23}\}, C_{12}   \bigr] + 4ab [Q,C_{12}]\, .
  \end{equation}  
  It is seen on the other hand that
  \begin{equation}
\{ Q_{12}, Q_{13} \}=\{ Q_{13}, Q_{23} \}=\{ Q_{12}, Q_{23} \}
  \end{equation}  
  and therefore that
  \begin{equation}
   Q^2= 3 +3 \{ Q_{12}, Q_{13} \} \, .
  \end{equation}  
  This confirms that $C$ commutes with $C_{12}$.  The same goes for  $C_{13}$ and  $C_{23}$. Consider now the commutator $[C,C_1]$.
  Note that
    \begin{equation}
      \begin{split}
        [ C_{12}^2 +  C_{13}^2 +  C_{23}^2, C_{1}]& = a  [ C_{12}^2 +  C_{13}^2 +  C_{23}^2, Q_{12}+Q_{13}] \\
        & =  a [ C_{13}^2 +  C_{23}^2, Q_{12}]+ a [ C_{12}^2  +  C_{23}^2, Q_{13}]
        \end{split}
    \end{equation}  
    and check that each commutator in the last line vanishes because of \eqref{eq5.25}.  Observe moreover that
    \begin{equation}
[Q,C_1] = a[Q_{13},Q_{12}+Q_{23}]=0 \, .
    \end{equation}  
    It follows that $[C,C_1]=0$.  One shows similarly that  $[C,C_2]=[C,C_3]=0$.
\end{proof}  
\begin{remark}
  We may expect the invariant $C$ to be related to the Casimir operator of $\osp(1,2)$ and indeed one finds that
  \begin{equation}
C= \Gamma^2+ 3(a^2+b^2)-\frac{1}{4}\,.
  \end{equation}  
\end{remark}

\section{$\osp(1,2)$ realizations with Clifford algebras} \label{sec6}
It is known that $\osp(1,2)$ can be realized using Dunkl operators and generators of Clifford algebras~\cite{cit16}.
The centralizer of the superalgebra is then identified with the symmetries of the corresponding Dirac-Dunkl (massless) equation.
We now wish to indicate how the formalism developed in this paper applies in the Clifford algebra context and to give the expressions
of the operators that then form the generalized Bannai-Ito algebra.

Consider a set of $n$ generators $e_i$, $i=1,\dots,n$, of a Clifford algebra which verify
\begin{equation}
\{e_i,e_j \}=2 \delta_{ij} \, .
\end{equation}  
These $e_i$ are taken to commute with operators acting on functions of the coordinates $x_i$, $i=1,\dots,n$,
and in particular with permutations operators.  Let $\mathcal D_i$ be Dunkl operators associated to some arbitrary reflection group and take
\begin{equation} \label{eq6.2}
\tilde A_-= \sum_{i=1}^n \mathcal D_i e_i
\end{equation}  
and
\begin{equation} \label{eq6.3}
\tilde A_+= \sum_{i=1}^n  x_i e_i \, .
\end{equation}  
Given that $[\mathcal D_i,\mathcal D_j]=0$, it is manifest that
\begin{equation}
\tilde B_- = \tilde A_-^2= \sum_{i=1}^n \mathcal D_i^2
\end{equation}  
and
\begin{equation}
\tilde B_+ = \tilde A_+^2= \sum_{i=1}^n x_i^2 \, .
\end{equation}  
$\tilde B_-^2$ and $\tilde A_-$ are respectively the Laplace-Dunkl and Dirac-Dunkl operators.  It can be checked \cite{cit16} that the operators
$\tilde A_{\pm}$ realize the $\osp(1,2)$ commutation relations $\{\tilde A_-,\tilde A_+  \}= 2 \tilde A_0$,  $[\tilde A_0,\tilde A_{\pm}  ]
= \pm \tilde A_{\pm}$, when $\tilde A_0$ is taken to be the Euler operator
\begin{equation} \label{eq6.6}
\tilde A_0 = \frac{1}{2} \sum_{i=1}^n \{ x_i,\mathcal D_i\} \, .
\end{equation}
Provided $\{R_i,\mathcal D_i\}=0$, $[R_i, \mathcal D_j ]=0$,
    $i\neq j$, the grade involution can again be realized by $R=\prod_{i=1}^n R_i$.
As already observed, this last equation is enforced for the Dunkl operators associated to the reflection groups $\mathbb Z_2^n$ and $B_n$.
In the case of $\mathbb Z_2^n$, albeit in a model with Clifford generators, the coproduct embedding of $\osp(1,2)$ into
$\osp(1,2)^{\otimes 3}$ is realized and the centralizer will in general be the higher rank Bannai-Ito algebra \cite{cit10}.
We shall therefore focus anew on the $B_n$-Dunkl operators $D_i$ given in \eqref{eq5.3} and take $n=3$.  Replacing $\mathcal D_i$ by $D_i$
in \eqref{eq6.2}-\eqref{eq6.3}, we now observe that 
\begin{equation}
  [\tilde A_-,R_i] = 2 D_i e_i R_i, \qquad [\tilde A_+,R_i] = 2 x_i e_i R_i  \, .
\end{equation}  
The requirement $\Bigl[ R_i, [R_j , \tilde A_{\pm}] \Bigr]=0$, $i\neq j$, is thus satisfied.  We can therefore apply our formalism here also to compute the generators of the corresponding centralizer in this model of $\osp(1,2)$.
\begin{proposition} \label{prop6.1}
  Given the realization \eqref{eq6.2}, \eqref{eq6.3} and \eqref{eq6.6}, for $n=3$ the elements $C_i$, $C_{ij}$ and the Casimir operator
  $C_{ijk}=\Gamma$ that satisfy the relations \eqref{eq3.33} and \eqref{eq3.34} of the generalized Bannai-Ito algebra are given by
  \begin{equation} \label{eq6.8}
C_i = \frac{1}{2} \bigl(S_{ii}-S_{ij} e_i e_j -S_{ik}e_i e_k -1\bigr)R_i
  \end{equation}  
 \begin{equation} \label{eq6.9}
   C_{ij} = -M_{ij} e_i e_jR_i R_j+
   \frac{1}{2} \bigl(S_{ii}+S_{jj}-S_{ik} e_i e_k -S_{jk}e_j e_k -1\bigr)R_i R_j
 \end{equation}
 and
 \begin{equation}
   \Gamma = \left(  -M_{12} e_1 e_2  -M_{13} e_1 e_3  -M_{23} e_2 e_3
   +  \frac{1}{2} \bigl(S_{11}+S_{22}+S_{33}-1  \bigr) \right) R
 \end{equation}
 with $R=R_1 R_2 R_3$ and $i,j,k \in \{ 1,2,3 \}$ all distinct.
 \end{proposition}  
\begin{proof}
  These formulas are obtained by an explicit evaluation of
  $$C_S=\frac{1}{4}\Bigl\{ \tilde A_-, [\tilde A_+, R_S]  \Bigr\}-\frac{1}{2}R_S$$
for $S=\{ i\},\{ i,j\},\{ i,j,k\}$.
\end{proof}  
\begin{remark}
  It is readily seen that the expression for $\Gamma=C_{123}$ is in keeping with \eqref{eq3.26} and that we have
  \begin{equation}
\Gamma = C_{12}R_3 + C_{13}R_2 + C_{23}R_1 - C_1 R_2 R_3 -  C_2 R_1 R_3  -C_3 R_1 R_2 -\frac{1}{2} R_1 R_2 R_3 \, .
  \end{equation}  
  \end{remark}  
\begin{remark} When $a=0$, the generators become
  \begin{equation}
      C_{ij} =     \Bigl( -M_{ij} e_i e_j  +b(R_i + R_j) +\frac{1}{2}  \Bigr) R_i R_j, \qquad
      C_i  = b, \qquad 
      D_i  = \frac{\partial}{\partial x_i} + \frac{b}{x_i}(1-R_i)\, .
  \end{equation}  
We then recover the symmetries of the Dirac-Dunkl equation found in \cite{cit9} with the parameters of the $\mathbb Z_2^3$ Dunkl operators all equal to $b$.
\end{remark}  
Let us record another expression for $C_i$.
\begin{proposition} The 1-index elements $C_i$ can be given as follows
  \begin{equation}
C_i = a (W_{ij}+W_{ik}) e_i R_i +b,\qquad i,j,k \in \{1,2,3 \} \text{ all distinct}
  \end{equation}  
  with 
     \begin{equation}
W_{ij}= \frac{1}{2} \Bigl( (e_i-e_j)\pi_{ij} + (e_i+e_j) R_i R_j \pi_{ij}  \Bigr)\, .
  \end{equation}
  \end{proposition}  
  \begin{proof} This result follows from \eqref{eq6.8} and the definition \eqref{eq5.16}.  It is the analog of \eqref{eq5.27}. Notice that $W_{ij}$
    has mixed symmetry.
  \end{proof}  

  We now wish to discuss the relation that the results of this section have with the study of the symmetries of the Dirac-Dunkl equations carried in \cite{cit11}.  Let us stress that the grade involution $P$ and its decomposition into $P = \prod_{i=1}^n P_i$ are central in our
  abstract framework.  In contradistinction, working exclusively in the realm of Clifford algebras, the authors of \cite{cit11} have designed an alternative method to obtain symmetries of Dirac-Dunkl equations specifically.  Let us briefly review their approach which has the merit of applying to any reflection group.

  Take $[n]=\{1,\dots,n \}$ and $S=\{s_1,\dots,s_k  \} \subseteq [n]$.  In the notation of Section~\ref{sec3}, let $x_S= \sum_{i=1}^k x_{s_i} e_{s_i}$,
  $\mathcal D_S= \sum_{i=1}^k \mathcal D_{s_i} e_{s_i}$ and $e_S= \prod_{i=1}^k  e_{s_i}$. Script letters will be used in the following to identify quantities that pertain to the generic Dunkl operator $\mathcal D_i$.  As shown in \cite{cit11}, the quantities
  \begin{equation}
\mathcal O_S= \frac{1}{2} \Bigl( \mathcal D_{[n]} x_S e_S -e_S x_S  \mathcal D_{[n]} -e_S  \Bigr)
  \end{equation}  
  either commute or anticommute with $\mathcal D_{[n]}$ and  $x_{[n]}$ for all $S$, depending on the cardinality $|S|$ of $S$; namely,
  \begin{equation} \label{eq6.16}
    \mathcal D_{[n]} \mathcal O_S  = (-1)^{|S|}   \mathcal O_S \mathcal D_{[n]} \quad \text{and} \quad
     x_{[n]} \mathcal O_S = (-1)^{|S|}     \mathcal O_S \, x_{[n]}\, .
  \end{equation}  
The quantities $\mathcal O_S$ and the algebra they form are the objects of consideration in \cite{cit11}.  Our $C_S$ and their algebraic properties are a priori distinct.

Now, for $n=3$, $i,j,k \in \{ 1,2,3\}$ and unequal, a simple application of formula \eqref{eq6.16} gives
\begin{equation} \label{eq6.17}
  \mathcal O_{ij} = \mathcal M_{ij} + \frac{1}{2}\bigl(\mathcal S_{ii}+ \mathcal S_{jj}\bigr) e_i e_j-\frac{1}{2} \mathcal S_{ik}e_j e_k
  +\frac{1}{2} \mathcal S_{jk}e_i e_k -\frac{1}{2}e_i e_j 
\end{equation}  
with
\begin{equation}
 \mathcal M_{ij} = x_i \mathcal D_j- x_j \mathcal D_i \quad \text{and} \quad \mathcal S_{ij} = [x_i, \mathcal D_j] \, .
\end{equation}  
This operator $\mathcal O_{ij}$, according to \eqref{eq6.16}, will commute with $\tilde A_-=\mathcal D_{[n]}$ and  $\tilde A_+=x_{[n]}$.  As it happens, when the $\mathcal D_i$ are the $B_3$-Dunkl operators $D_i$, this would seemingly give centralizing elements that differ from those we have already found and given in Proposition~\ref{prop6.1}.  This is reconciled by remarking that there are additional centralizing elements or symmetries when
the Dunkl operators at hand transform like derivatives under the reflections, a case in point for the $B_n$-Dunkl operators.  When this is so, it is readily verified \cite{cit9} that the operators
\begin{equation}
Z_i= e_i R_i
\end{equation}  
will obey
\begin{equation}
\{ \tilde A_{\pm}, Z_i\}=0 \, .
\end{equation}  
Products $Z_i Z_j$ will hence be symmetries.  Focus now on quantities $\mathcal O_S$ associated to the $B_3$-Dunkl operators $D_i$.  In light
of the previous observation, upon comparing \eqref{eq6.9} and \eqref{eq6.17}, we find that
\begin{equation}
C_{ij} = \mathcal O_{ij} e_i e_j R_i R_j,
\end{equation}  
in other words, we see that $C_{ij}$ and $\mathcal O_{ij}$ differ by a factor which is itself a symmetry given that the $B_3$-Dunkl operators
satisfy \eqref{eq5.6}.

Similarly, we can show that
\begin{equation}
C_i= \mathcal O_i e_i R_i \, .
\end{equation}  
Therefore, by multiplying two quantities $\mathcal O_i$ and $e_i R_i$ that anticommute with $\tilde A_{\pm}$, we find the
centralizing $C_i$.  This explains the relation between the two approaches when our formalism is applied to $\osp(1,2)$
realizations with Clifford algebras.

\section{Conclusion}

This paper has introduced generalizations of the Bannai-Ito algebra that are of rank 1.  They are intimately connected to the Lie superalgebra $\osp(1,2)$. It has been shown that whenever the grade involution $P$ of $\osp(1,2)$ factors into a product of appropriate involutions, elements that centralize $\osp(1,2)$ appear.

When $P$ admits three such factors, $P=P_1 P_2 P_2$,
the centralizer thus formed generalizes the Bannai-Ito algebra.  Realizations of $\osp(1,2)$ in terms of Dunkl operators have been considered to obtain concrete models of generalized Bannai-Ito algebras.  This has produced in particular an hyperoctahedral extension of the Bannai-Ito algebra that exhibits an interesting structure.

This study brings up many questions.  Determining which generalized Bannai-Ito algebras other realizations of $\osp(1,2)$ would entail
naturally comes to mind.  What would the approach bring when applied to other superalgebras or more complicated structures is also among those questions.  The number of supplementary involutions has been restricted to three in order to obtain centralizers of rank 1. Higher rank cases will result upon lifting this constraint.  As a matter of fact, the hyperoctahedral generalization of arbitrary rank will be presented in a forthcoming publication \cite{cit21} from the perspective of the rational $B_n$-Calogero model for non-identical particles.  Developing the representation theory of this $B_n$-extended Bannai-Ito algebra would be instructive.  One expects relations with special functions in one and many variables extending the connection that the Bannai-Ito polynomials have with the eponym algebra.

We look forward to exploring these avenues in the near future.

\appendix

\section{}

We here give more details on how Theorem~\ref{theo3.7} is proven along the strategy described in Section~\ref{sec3} after the statement of the theorem.  Generally, as one proceeds with the proofs, it is appropriate to use the following expanded form 
\begin{equation} 
C_S = \frac{1}{4} \Bigl(  A_- A_+ P_S - A_-P_S A_+ + A_+ P_S A_-- P_S A_+ A_-  \Bigr)  -\frac{1}{2} P_S
\end{equation}  
for the centralizing elements.

\noindent Sketch of the proof of \eqref{eq3.33}.

Consider the l.h.s. of \eqref{eq3.33}.  In expanding $\{C_{ij}, C_{jk} \}$, the first thing to do is to eliminate the involutions $P_j$ with index $j$
by using formulas such as \eqref{eq3.23} to bring the two $P_j$ factors together and make them disappear given that $P_j^2=1$.  It is immediate to check that in \eqref{eq3.33}, the terms without $A_{\pm}$ involving only $P_i P_k$ match on both sides.

Let us now focus on the terms that are bilinear in $A_+$ and $A_-$.  In the anticommutators, these occur in the cross terms with the pure involution part of the relevant $C_S$ as a factor.  Start with the r.h.s. Using $\{ P,A_{\pm} \}=0$ and formula \eqref{eq3.23} repeatedly, one finds that the bilinear terms in $C_{ik}+\{C_{j}, C_{ijk} \} + \{C_{i}, C_{k} \}$ simplify to    
\begin{equation} 
  \begin{split}
    \text{bilinears on r.h.s. }& = \frac{3}{8} \bigl[P_i P_k, A_- A_+\bigr]-  \frac{1}{8} \bigl[P_i P_k, A_+ A_-\bigr] +
    \frac{1}{4} A_+ P_i P_k A_- \\
& \qquad     -  \frac{1}{4} A_- P_i P_k A_+ + \frac{1}{8} P_i [A_+,A_-] P_k   + \frac{1}{8} P_k [A_+,A_-] P_i \, .
  \end{split}  
\end{equation}  
This is seen to coincide with what is found after reducing with the same tools the sum of the bilinears in $\{C_{ij}, C_{jk} \}$, the l.h.s.
Remaining is the part that is quartic in the odd generators.  Let us describe for example how one deals with the terms of the form
$A_+ A_-^2 A_+$. Recall that $[P_i, A_-^2]=[P_j, A_-^2]= [P_k, A_-^2]=0$.  After the elimination of $P_j$, the terms featuring $A_-^2$ in 
$\{C_{ij}, C_{jk} \}$ are found to be
\begin{equation} \label{eqA.3}
  -\frac{1}{16} \Bigl( A_+ A_-^2 P_i A_+ P_k+  A_+ A_-^2 P_i P_k A_+ + P_i A_+ A_-^2  A_+ P_k +  P_i A_+ A_-^2 P_k A_+ + (i \leftrightarrow j)
  \Bigr)\, .
\end{equation}  
The terms in $A_-^2$ in $\{C_{i}, C_{k} \}$ are easily read off and found to coincide with those of $\{C_{ij}, C_{jk} \}$ given in
\eqref{eqA.3} except that the first and last have opposite signs.  This is nicely compensated by the corresponding terms in
$\{C_{j}, C_{ijk} \}$.  Indeed, collecting the terms in $A_-^2$ in
\begin{equation} 
  \begin{split}
    &  \{C_{j}, C_{ijk} \}  \\
    & \, \,  = \left\{ \frac{1}{4} \bigl(  A_- A_+ P_j -  A_- P_j A_+ +  A_+ P_j A_- -  P_j A_+ A_-   \bigr) - \frac{1}{2}P_j,
  \frac{1}{2} \bigl( A_- A_+ - A_+ A_- \bigr) P - \frac{1}{2} P
  \right\}
  \end{split}
\end{equation}  
yields
\begin{equation} 
  \begin{split}
     \frac{1}{8} \Bigl( - A_+ P_i P_k A_-^2 A_+  -   P_i P_k A_+ & A_-^2 A_+  -  A_+  A_-^2 A_+ P_i P_k- A_+ A_-^2 P_i P_k  A_+ 
    \Bigr) \\
& \qquad =  -\frac{1}{8} \Bigl(  A_+ A_-^2 P_i  A_+ P_k +  P_i A_+ P_k A_-^2 A_+ + (i \leftrightarrow j)   \Bigr)     \, .
  \end{split}  
\end{equation}  
It thus follows that the terms in $A_-^2$ in  $\{C_{j}, C_{ijk} \} + \{C_{i}, C_{k} \}$ are equal to those of  $\{C_{ij}, C_{jk} \}$.

One then follows the same approach with the other types in the quartic part to complete the proof of relation \eqref{eq3.33}.
Note that when dealing with the $(A_-A_+)^2$ and $(A_+A_-)^2$ terms jointly, the ones found in $\{C_{j}, C_{ijk} \}$ read
\begin{equation} 
\begin{split}
&   \frac{1}{8} \Bigl(  A_- A_+ P_i P_k A_- A_+  +  A_- A_+ A_- P_i P_k  A_+  +  A_- P_i P_k A_+  A_- A_+  +  A_- A_+ A_- A_+ P_i P_k    \Bigr) \\
& \quad  +   \frac{1}{8} \Bigl(  A_+ P_i P_k A_-  A_+ A_-  +  A_+ A_-  P_i P_k A_+  A_-  + P_i P_k A_+  A_-  A_+ A_-  +  A_+ A_- A_+  P_i P_k A_-    \Bigr) \, .
  \end{split}
\end{equation}  
The identity \eqref{eq2.15}, as well as \eqref{eq2.14}, are here needed to rewrite those terms in the form
\begin{equation} 
  \frac{1}{8} \Bigl(  A_- A_+ P_i  A_- P_k A_+  +  A_- P_i A_+ A_-   A_+ P_k   + A_+ P_i A_- P_k  A_+ A_-  +  P_i A_+ A_- A_+  P_k A_-  +
  (i \leftrightarrow k)   \Bigr)\, .
\end{equation}  
They are then seen to perfectly combine with the analogous terms in  $\{C_{i}, C_{k} \}$ to equal what is found in  $\{C_{ij}, C_{jk} \}$.

\noindent Sketch of the proof of \eqref{eq3.34}.

It is immediate to see that the terms involving only the involutions trivially vanish.  Recall that $C_{ij}= \frac{1}{4} \bigl \{A_-, [A_+, P_i P_j]  \bigr\} -\frac{1}{2} P_i P_j$.  From Lemma~\ref{lem3.5} and its specialization as formula \eqref{eq3.24} for $n=3$, we see that the bilinear part in $A_-$ and $A_+$ of \eqref{eq3.34} reduces to
\begin{equation} 
  -\frac{1}{8} \Bigl[ A_- A_+ P_i P_j  - A_-  P_i P_j A_+ + A_+  P_i P_j A_-  -   P_i P_j A_+ A_-, P_k  
    \Bigr] + \text{ cyclic permutations.} 
\end{equation}  
Expanding over the cyclic permutations and combining terms to use \eqref{eq2.14}, one arrives at
\begin{equation}
-\frac{1}{8} \Bigl[3  A_- A_+ P   + 3P  A_+  A_-  -  P_i P_j \{A_+,A_- \} P_k   -    P_j P_k \{A_+,A_- \} P_i - P_i P_k \{A_+,A_- \} P_j  
    \Bigr] 
\end{equation}  
which adds up to zero since the involutions commute with $\{A_+,A_- \}=2 \A_0$.

The quartic part is dealt with as in the proof of \eqref{eq3.33} by looking separately at the various orderings.  One regroups terms
arising from the different permutations and uses formula \eqref{eq2.14} to shift the involutions and observe that all terms cancel.

\begin{acknow} 
The authors wish to thank H. De Bie, R. Oste, J. Van der Jeugt  and W. van de Vijver for stimulating discussions.  VXG holds a postdoctoral fellowship from the Natural Science and Engineering Research Council (NSERC) of Canada.  The research of LL is supported by the Fondo Nacional de Desarrollo Cient\'{\i}fico y Tecnol\'ogico (FONDECYT) de Chile grant \#1170924.  LV gratefully acknowledges his support from NSERC through a discovery grant.  
\end{acknow}

\end{document}